%% file: main.tex
\documentclass[a4paper,UKenglish,cleveref, autoref, thm-restate]{lipics-v2021}

\pdfoutput=1 
\hideLIPIcs  

\graphicspath{{./images/}}

\title{Subcubic algorithm for (Unweighted) Unrooted Tree Edit Distance}

\author{Krzysztof Pióro}{Jagiellonian University, Kraków}{krzysztof.pioro1998@gmail.com}{}{}

\authorrunning{K. Pióro} 

\Copyright{Krzysztof Pióro} 



\category{} 

\relatedversion{} 

\acknowledgements{I would like to thank Adam Polak for introducing me to the topic, numerous discussions regarding the problem and his great help in editing this paper.}

\nolinenumbers 

\EventEditors{John Q. Open and Joan R. Access}
\EventNoEds{2}
\EventLongTitle{42nd Conference on Very Important Topics (CVIT 2016)}
\EventShortTitle{CVIT 2016}
\EventAcronym{CVIT}
\EventYear{2016}
\EventDate{December 24--27, 2016}
\EventLocation{Little Whinging, United Kingdom}
\EventLogo{}
\SeriesVolume{42}
\ArticleNo{23}

\usepackage{mathtools, listings}
\usepackage{algorithm, caption,subcaption,etoolbox,tikz,pgfplots}
\usepackage[noend]{algpseudocode}

\pgfplotsset{compat=1.7}
\newcommand\set[1]{\left\{#1\right\}}

\DeclareMathOperator\Oh{\mathcal{O}}
\DeclareMathOperator\Ot{\widetilde{\mathcal{O}}}
\DeclareMathOperator{\simi}{sim}
\DeclareMathOperator{\ed}{ed}
\DeclareMathOperator{\rangemax}{rangemax}
\DeclareMathOperator{\mincol}{mincol}
\DeclareMathOperator{\maxrow}{maxrow}
\DeclareMathOperator{\rot}{root}
\DeclareMathOperator{\mul}{MUL}
\DeclareMathOperator{\sub}{sub}

\makeatletter
\patchcmd{\ALG@doentity}{\item[]\nointerlineskip}{}{}{}
\makeatother

\begin{document}

\maketitle

\begin{abstract}
    The tree edit distance problem is a natural generalization of the classic string edit distance problem. Given two ordered, edge-labeled trees $T_1$ and $T_2$, the edit distance between $T_1$ and $T_2$ is defined as the minimum total cost of operations that transform $T_1$ into $T_2$. In one operation, we can contract an edge, split a vertex into two or change the label of an edge.

    For the weighted version of the problem, where the cost of each operation depends on the type of the operation and the label on the edge involved, $\Oh(n^3)$ time algorithms are known for both rooted and unrooted trees. The existence of a truly subcubic $\Oh(n^{3-\epsilon})$ time algorithm is unlikely, as it would imply a truly subcubic algorithm for the APSP problem. However, recently Mao (FOCS'21) showed that if we assume that each operation has a unit cost, then the tree edit distance between two rooted trees can be computed in truly subcubic time. 

    In this paper, we show how to adapt Mao's algorithm to make it work for unrooted trees and we show an $\Ot(n^{(7\omega + 15)/(2\omega + 6)}) \leq \Oh(n^{2.9417})$ time algorithm for the unweighted tree edit distance between two unrooted trees, where $\omega \leq 2.373$ is the matrix multiplication exponent. It is the first known subcubic algorithm for unrooted trees.

    The main idea behind our algorithm is the fact that to compute the tree edit distance between two unrooted trees, it is enough to compute the tree edit distance between an arbitrary rooting of the first tree and every rooting of the second tree. 
\end{abstract}

\input{paragraphs/introduction.tex}

\input{paragraphs/preliminaries.tex}

\input{paragraphs/cubic.tex}

\input{paragraphs/subcubic.tex}

\input{paragraphs/remarks.tex}

\bibliography{main}

\end{document}

%% file: paragraphs/introduction.tex
\section{Introduction}
\label{sec:introduction}
The tree edit distance problem is a natural generalization of the classic string edit distance problem. Given two trees $T_1$ and $T_2$, the edit distance between $T_1$ and $T_2$ is defined as the minimum total cost of operations that transform $T_1$ into $T_2$. The exact definition of the problem depends on whether we consider rooted or unrooted trees.

For the unrooted variant, which is the focus of this paper, the trees are edge-labeled and for each vertex its neighbors form a cyclic order. We can think of unrooted trees as if they were embedded in the plane. In one operation, we can contract an edge, split a vertex into two or change the label of any edge. The cost of each operation depends on the type of the operation and the label on the edge involved.

Computing the edit distance between two trees has found applications in many different areas, such as computational biology \cite{RNA, BIO2}, image processing \cite{IMG1, IMG2, IMG3, IMG4} and comparing XML data \cite{XML1, XML2, XML3}. One of the most notable examples is comparing the secondary structures of RNA molecules, which can be represented as rooted trees \cite{RNA}.

The tree edit distance problem has been already studied for many years. In 1977, Tai \cite{n6} showed the first algorithm that computes the edit distance between two rooted trees of size $n$ in $\Oh(n^6)$ time. Next, Zhang and Shasha \cite{n4} used a dynamic programming approach to improve the complexity to $\Oh(n^4)$ time. Their dynamic programming scheme was used as a basis for later algorithms. 
Klein \cite{Klein} showed how to optimize their algorithm to $\Oh(n^3 \log n)$ time by better choosing the direction of transitions in their dynamic programming scheme. 
Furthermore, he showed that his algorithm can be extended to also work for the unrooted trees in the same time complexity. 
Later, Demaine, Mozes, Rossman and Weimann \cite{n3} further improved the time complexity for the rooted case to $\Oh(n^3)$. Finally, Dudek and Gawrychowski \cite{Dudek} showed that the tree edit distance between unrooted trees can also be solved in $\Oh(n^3)$ time.

From the lower bound side, Bringmann, Gawrychowski, Mozes and Weimann \cite{bound} showed that the existence of a truly subcubic $\Oh(n^{3-\epsilon})$ time algorithm for the tree edit distance is unlikely. They proved that the existence of such an algorithm implies a truly subcubic algorithm for the All Pairs Shortest Paths problem (assuming an alphabet of size $\Theta(n)$) and $\Oh(n^{k(1-\epsilon)})$ time algorithm for finding a maximum weight $k$-clique (assuming a sufficiently large constant size alphabet).

However, all of these previous algorithms and lower bounds work when the costs of the operations are arbitrary. Recently, Mao \cite{Mao} showed that if we assume that each operation has unit cost, then the tree edit distance between two rooted trees can be computed in $\Oh(n^{2.9546})$ time. 
Since in the weighted setting it was possible to obtain the same time complexity for unrooted trees as for rooted trees, Mao posed the following open problem:
\begin{center}\emph{Is it possible to compute the unweighted tree edit distance between two unrooted trees in subcubic time?}\end{center}
We answer this question affirmatively.

\subsection{Our contribution}
We build on ideas of Mao \cite{Mao} and Klein \cite{Klein} and obtain the first ever known subcubic algorithm for the tree edit distance between unrooted trees in the unweighted setting. 
Our main result is the following: 

\begin{theorem}\label{th:main}
    There is an $\Ot(mn^{(5\omega + 9)/(2\omega + 6)}) \leq \Oh(mn^{1.9417})$ time algorithm that computes the (unweighted) tree edit distance between two unrooted trees of sizes $n$ and $m$.
\end{theorem}

We were not able to adapt one of the optimizations from Mao, thus our algorithm is slightly slower than the one for rooted trees. Note that, however the numerical value of the exponent in our result is actually better than in the original Mao's paper. It is only due to the fact that we replaced the algorithm for max-plus multiplication of two bounded-difference $n \times n$ matrices that Mao used 
(running in $\Oh(n^{2.8244})$ time) with a more recent result from Chi, Duan, Xie and Zhang ($\Oh(n^{2.687})$ time) \cite{Matrix}. Dürr \cite{Durr} showed that with this new result, Mao's algorithm works in $\Ot(mn^{1.915})$ time. Thus, our algorithm is in fact slower than Mao's algorithm.

\subsection{Technical overview}
\subparagraph*{Sketch of Mao's algorithm for rooted trees: }
First, let us note that in the unrooted tree edit distance, we consider edge-labeled trees, but Mao's algorithm works for node-labeled trees. However, the modification of Mao's algorithm to work for rooted edge-labeled trees is simple. Indeed, for both trees we can introduce a virtual root that becomes a parent of the original root and then we can associate the label of every vertex with the label of edge to its parent. 

Mao's algorithm is based on Chen's algorithm \cite{Chen}. Chen showed a completely different approach than previous algorithms that were based on Zhang and Shasha's dynamic programming scheme. He reduced the problem to the $(\min, +)$ matrix multiplication and obtained an algorithm working in $\Oh(n^4)$ time\footnote{By reduction we mean that he used $(\min, +)$ matrix multiplication as a sub-procedure in his algorithm.}. 

Mao builds on Chen's approach but considers an equivalent maximization problem of the similarity of trees. Due to this, the matrices occurring in his algorithm satisfy additional properties: they are monotone and the difference between two adjacent cells is bounded by a constant. For such matrices, there exists a truly subcubic time algorithm that computes the $(\min, +)$ product \cite{Matrix}, which is crucial for obtaining a subcubic complexity in the tree edit distance problem.

In summary, Mao's algorithm, given two trees $T_1$ and $T_2$, computes matrices $S(T_1')$ for some subtrees $T_1'$ of tree $T_1$. Every such matrix $S(T_1')$ encodes the similarity between tree $T_1'$ and all relevant fragments of tree $T_2$. These fragments are defined as segments of the Euler tour sequence\footnote{Actually, Mao used the so-called bi-order traversal sequence, but Euler tour sequence is its equivalent for edge-labeled trees.} and are formally defined in Definition \ref{def:segment}. 

Mao first shows a dynamic programming scheme based on Chen's algorithm that computes $S(T_1)$ in $\Oh(n^4)$ time. Next, he optimizes it to $\Oh(n^3)$ time by exploiting the special properties of similarity matrices. Then, to achieve a subcubic complexity, he presents a special decomposition scheme, which allows him to skip some of the subproblems. For block size $\Delta = n^d$ for some $d$ slightly smaller than $0.5$, he decomposes the computation of $S(T)$ into $\Oh(n/\Delta)$ transitions of one of two types. 

The first type is a concatenation of two trees. For that, he shows an $\Oh(t^{1-\epsilon} n^2)$ time algorithm that computes the $(\max, +)$ product of two bounded-difference matrices such that the entries in one of them are bounded by $t$. The second transition type involves expanding a subtree by adding a path going up from its root along with some additional subtrees attached to this path. For these transitions, Mao presents a special three-part combinatorial method.

\subparagraph*{Sketch of our algorithm for unrooted trees: }
To generalize Mao's algorithm to unrooted trees, we use the idea that Klein \cite{Klein} used in his $\Oh(n^3 \log n)$ time algorithm. Klein used the fact that to compute the tree edit distance between two unrooted trees, it is enough to consider arbitrary rooting of the first tree and try
all possible rootings of the second tree. 

Direct application of this fact requires solving $\Oh(n)$ particular instances of the rooted tree edit distance. However, Klein showed that his $\Oh(n^3 \log n)$ time algorithm for the rooted tree edit distance can be modified to solve all of these instances at once in the same time complexity. 
To achieve this, he used the fact that his algorithm for rooted trees computes the edit distance between some fragments of the first tree and all subtrees of the second tree that correspond to some segment of the Euler tour sequence of that tree. Since the different rootings of the second tree correspond to different cyclic shifts of the Euler tour sequence, Klein modified his algorithm so that it computes the edit distance between some fragments of the first tree and all cyclic segments of the Euler tour sequence of the second tree.

We notice that Mao's algorithm has properties similar to those of Klein's algorithm. For some of the subtrees $T_1'$ of the first tree, it computes a similarity matrix $S(T_1')$ that encodes the similarity between tree $T_1'$ and all subtrees of $T_2$ that correspond to some segment of the Euler tour sequence of that tree.  Thus, we need to modify the algorithm to handle cyclic segments of the Euler tour sequence.  To do this, we build our new similarity matrices on a doubled Euler tour sequence (see Definitions 7--9 in Section 2). Now every cyclic segment of the original Euler tour occurs as a normal segment in our sequence. 

The introduction of a doubled Euler tour requires us to make a few modifications to Mao's algorithm. The key technical challenge is handling of multiplication of new similarity matrices which we describe in Section \ref{sec:mat_mul}.

\subsection{Organization}
In Section \ref{sec:prelim} we present basic definitions and notation. We also describe there the reduction from the unrooted tree edit distance to the rooted tree edit distance and we formally define similarity matrices for unrooted trees. 

Then in Section \ref{sec:basicdp} we define our basic dynamic programming scheme that computes the similarity of trees and we show how to compute it in $\Oh(n^3)$ time. The idea behind this scheme is the same as in Mao's algorithm and we only needed to do some simple changes to adapt it to unrooted trees.

Next in Section \ref{sec:subcubic} we show how to optimize our algorithm to subcubic time using Mao's decomposition scheme. The decomposition procedure described in Section \ref{sec:decomp} is exactly the same as the one used by Mao.
Our main technical contribution is showing how to generalize Mao's matrix multiplication algorithm, which is used for the type I transition, to similarity matrices of unrooted trees. We show this in Section \ref{sec:mat_mul}. Then in Section \ref{sec:type2} we show how to implement type II transitions. This implementation is based on the one given by Mao, with some simple changes to adapt it to unrooted trees.
Finally, in Section \ref{sec:time}, we analyze the total running time in terms of :
\begin{itemize}
    \item complexity of the $(\max, +)$ matrix multiplication of similarity matrices of unrooted trees,
    \item complexity of the $(\max, +)$ matrix multiplication of bounded-difference matrices,
    \item $\omega$, the fast matrix multiplication exponent.
\end{itemize}   

%% file: paragraphs/preliminaries.tex
\section{Preliminaries}
\label{sec:prelim}
In the tree edit distance problem we consider ordered trees with labels on edges.  
We consider both rooted and unrooted trees.
For rooted trees, ordered means that for each vertex we are given a left-to-right order of its children.
For unrooted trees, ordered means that for each vertex its neighbors form a cyclic order.  

Note that we can transform an unrooted tree into a rooted tree by choosing a root and the first edge to its children. 
This choice carries over to the other vertices and defines the order of their children.
It uniquely determines the rooting of the tree, thus there are $2(n-1)$ possible rootings for an unrooted tree with $n$ vertices.

For simplicity, we first assume that both input trees are of equal size $n$, but at the end we address the case when they have different sizes.

As we are interested in unrooted tree edit distance, we consider only trees that are edge labeled
(for rooted tree edit distance it is more common to have labels on nodes). 
In this paper, we assume that every edge is labeled from some alphabet $\Sigma$ of size $\Oh(n)$.

\begin{definition}[{(Unweighted) Tree Edit Distance}]
    Let $T_1$ and $T_2$ be rooted, ordered trees with labels on edges. We consider two types of operations:
    \begin{itemize}
        \item Label change of a selected edge in tree $T_1$ or $T_2$.
        \item Contraction of a selected edge in tree $T_1$ or $T_2$. 
        When contracting an edge $pu$ where $p$ is parent of $u$, children of $u$ become children of $p$, they replace $u$ in the children's list of $p$ and keep their order. 
    \end{itemize}
    \textit{The tree edit distance} between $T_1$ and $T_2$, denoted by $\ed(T_1, T_2)$, is the 
    minimum number of operations we have to perform on $T_1$ and $T_2$ to transform both trees into an identical tree.
\end{definition}

\begin{definition}[{(Unweighted) Unrooted Tree Edit Distance}]
    Let $T_1$ and $T_2$ be unrooted, ordered trees with labels on edges.
    We define \textit{the tree edit distance} between $T_1$ and $T_2$ 
    as the minimum edit distance over all possible rootings of $T_1$ and $T_2$.
\end{definition}

Klein \cite{Klein} mentioned, it is enough to consider arbitrary rooting of the first tree 
and try all possible rootings of the second tree to find an optimal solution for the unrooted tree edit distance. 
This should be intuitively evident, however we will include here a short proof
of that fact for completeness. 
\begin{lemma}
\label{reduction}
    Let $T_1$ and $T_2$ be unrooted trees. For every rooting of tree $T_1$ there is at least one rooting of $T_2$ 
    that admits minimum edit distance between $T_1$ and $T_2$.
\end{lemma}

\begin{proof}
    Let $T^*$ denote the common tree into which both trees $T_1$ and $T_2$ are transformed in an optimal solution. Consider any rooting of tree $T_1$. It induces a rooting of $T^*$. 
    However, it is easy to see that any rooting of $T^*$ can be extended to some rooting of tree $T_2$.
    Thus, edit distance between these rootings is optimal among all possible rootings of $T_1$ and $T_2$.
\end{proof}

Same as Mao, we consider a maximization problem equivalent to the tree edit distance problem. 

\begin{definition}[{Similarity}]
    The \textit{similarity} between two rooted trees $T_1$ and $T_2$ is defined as $\simi(T_1, T_2) = |E(T_1)| + |E(T_2)| - \ed(T_1, T_2)$.
\end{definition}
Similarity can also be interpreted as the weight of the heaviest matching between the edges of the tree $T_1$ and $T_2$, where the cost of edge matching is $2$ when edges have the same labels and $1$ when they are different. In addition, the matching must respect the tree structure, which means:
\begin{itemize}
    \item if edge $a \in T_1$ is matched to edge $b \in T_2$, then edges in the subtree of $a$ can be matched only to edges in the subtree of $b$,
    \item if edge $a \in T_1$ is matched to edge $b \in T_2$ and $c \in T_1$ is matched to $d \in T_2$, then $a$ is ``to the left'' of $c$ if and only if $b$ is ``to the left'' of $d$.
\end{itemize}
Note that $0 \leq \simi(T_1, T_2) \leq 2 \min(|E(T_1)|, |E(T_2)|)$.

\begin{definition}
    By $\eta(e, f)$ we denote the cost of matching edges $e$ and $f$, that means $\eta(e, f)=2$ if labels of these edges are equal and $\eta(e, f)=1$ otherwise. 
\end{definition}

\subparagraph*{Tree definitions:} For a rooted tree $T$, by $L_T$ we denote the subtree of the first (leftmost) child of the root of $T$ along with the edge to the root of $T$. Similarly, by $R_T$ we denote the subtree of the last (rightmost) child of the root of $T$ along with the edge to the root of $T$. Given an edge $e$ in a rooted tree, we use $sub(e)$ to represent the subtree rooted at edge $e$.

For two rooted trees $T_1$ and $T_2$, by $T_1 + T_2$ we denote the tree formed by merging the roots of tree $T_1$ and $T_2$ such that edges from the tree $T_1$ are ``to the left'' of the edges from the tree $T_2$. 

For two trees (rooted/unrooted) $T_1$ and $T_2$ such that $T_1 \subseteq T_2$, by $T_2 - T_1$ we denote the tree formed from tree $T_2$ by contracting all edges that appear in tree $T_1$.

To describe the ``fragments'' of a tree that we will consider in our algorithm, we use segments of the Euler cycle of the tree. 
Due to technical reasons, instead of dealing with a cyclic sequence, we consider a doubled
Euler tour sequence.

\begin{definition}
    Let $T$ be an unrooted tree. Consider the walk on this tree that starts at an arbitrary edge and goes twice through the Euler Tour, which visits neighbors according to their order.
    \begin{enumerate}[(a)]
        \item For $i \in \set{1,\ldots, 4|E(T)|}$ by $T(i)$ we denote the $i$-th edge of this walk. 
        \item By $I(e)_i$ we denote the index of the $i$-th occurrence of the edge $e$ in this walk.
        \item By $l_{i,j}(e)$ we denote the index of the first occurrence of the edge $e$ in $T(i),\ldots, T(j-1)$. 
        \item By $r_{i,j}(e)$ we denote the index of the second occurrence of the edge $e$ in $T(i),\ldots, T(j-1)$. 
    \end{enumerate}
\end{definition}

Note that each edge appears exactly 4 times in this walk.

\begin{definition}[Segment]\label{def:segment}
    Given an unrooted tree $T$ and integers $l$, $r$ such that $1\leq l \leq r \leq 4|E(T)|$ and $r - l \leq 2|E(T)|$, by $T[l, r)$ we denote the tree formed from the tree $T$ by contracting every edge that occurs less than $2$ times in $T(l),\ldots,T(r-1)$.
\end{definition}

Let us note here that despite the fact that our input tree is unrooted,
we can view the segments of this tree as rooted trees. Indeed, the first non-contracted edge
that appears in the segment defines the rooting of this segment.
Thus, the segments of length $2|E(T)|$ define all possible rootings of the tree $T$. See Figure \ref{fig:euler}.

\begin{figure}[!h]
    \begin{center}
    \includegraphics[scale=1.3]{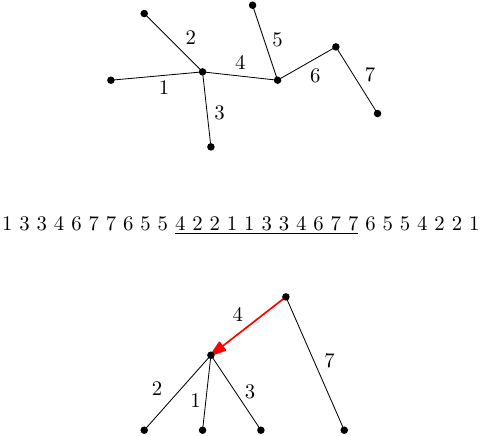}
    \caption{Example unrooted tree $T$ with its doubled Euler tour sequence and segment $T[11,22)$ of this tree. The edge marked as red defines the rooting of this segment.}\label{fig:euler}
    \end{center}
\end{figure}

\subparagraph*{Matrix definitions:}
Now we are ready to define similarity matrix -- it encodes information about similarity of the
whole rooted tree $T$ and all ``relevant fragments'' of unrooted tree $Q$.
\begin{definition}[Similarity matrix]
    Given a rooted tree $T$ and an unrooted tree $Q$ \textit{similarity matrix} $S(T, Q)$ is a matrix of size 
    $4|E(Q)| \times 4|E(Q)|$, where:
    $$
			S(T,Q)_{i,j} =
				\begin{cases}
					\simi(T, Q[i,j)) & \text{ if } i \leq j \text { and } j - i \leq 2|E(Q)| \\
					-\infty & \text{ if } i > j \text { or } j - i > 2|E(Q)|
				\end{cases}
    $$
\end{definition}
For simplicity, we also introduce additional notation for naming of similarity matrix cells.
We divide them into three types:
\begin{itemize}
    \item \textit{valid cells} -- cells $S(T, Q)_{i, j}$ for which $i \leq j$ and $j - i \leq 2|E(Q)|$,
    \item \textit{invalid cells} -- cells $S(T, Q)_{i, j}$ for which $j - i > 2|E(Q)|$,
    \item \textit{cells under diagonal}.
\end{itemize}
During our algorithm we will allow invalid cells to store values other than $-\infty$.
Because of that, we introduce a special equality relation for similarity matrices:
\begin{definition}
    By $=_{valid}$ we denote the relation on matrices which is true if and only if they are equal on the set of valid cells. 
\end{definition}

\begin{definition}
    Let $A$ and $B$ be $n \times n$ matrices. By $A \star B$ we denote their $(\max, +)$ product i.e. $(A \star B)_{i, j} = max_{1 \leq k \leq n} A_{i, k} + B_{k, j}$.
\end{definition}

Now, we define a few properties of matrices. Given an $n \times n$ matrix $A$ we call it:
\begin{enumerate}[(a)]
    \item \textit{finite-upper-triangular} matrix if all entries below diagonal are $-\infty$ and the rest are finite,
    \item \textit{row-monotone} matrix if $A_{i,j} \leq A_{i, j+1}$ for all $i,j$,
    \item \textit{column-monotone} matrix if $A_{i+1,j} \leq A_{i, j}$ for all $i,j$,
    \item \label{bound} \textit{$W$-bounded-difference} matrix if for all $i, j$ we have:
    \begin{gather*}
        |A_{i,j} - A_{i-1, j}| \leq W \\
        |A_{i,j} - A_{i, j+1}| \leq W
    \end{gather*} 
    \item \textit{bounded-difference} matrix if it is a $W$-bounded-difference matrix for some $W = \Oh(1)$,
    \item \textit{finite-upper-triangular-$W$-bounded-difference} matrix if it is a finite-upper-triangular matrix and property \ref{bound} holds for all $i \leq j$,
    \item \textit{$M$-bounded similarity matrix} if it is a similarity matrix and all finite entries are integers between $0$ and $M$.
\end{enumerate}

It is easy to see that for similarity matrices $S(T, Q)$ the properties row-monotone and column monotone hold for all cells except invalid cells. The following lemma, which is an analogy to Lemma 4.1 from Mao \cite{Mao}, tells that the finite-upper-triangular-$2$-bounded-difference property also holds for all cells except invalid cells.
\begin{lemma}
    Given a rooted tree $T$ and an unrooted tree $Q$, we have:
    \begin{itemize}
        \item $\simi(T, Q[i,j+1)) \leq \simi(T, Q[i, j)) + 2$ for all $i\leq j, j-i+1 \leq 2|E(Q)|$.
        \item $\simi(T, Q[i-1,j)) \leq \simi(T, Q[i, j)) + 2$ for all $i \leq j, j-i+1 \leq 2|E(Q)|$.
    \end{itemize} 
\end{lemma} 
To prove this lemma it is enough to note that when we extend a tree segment by one we can get at most one additional edge which can contribute at most $2$ to the cost of the matching.

%% file: paragraphs/cubic.tex
\section{Cubic algorithm for Unrooted Tree Edit Distance}
\label{sec:basicdp}

In this section, we introduce a basic dynamic programming scheme that computes the similarity of trees. 
It is the same scheme that Mao used, except that we use slightly different similarity matrices. 
We show how to efficiently compute this dynamic programming scheme in $\Ot(n^3)$ time. This gives us a basic algorithm that will be later used in our subcubic algorithm.  

As mentioned in Lemma \ref{reduction}, it is enough to consider an arbitrary rooting of 
the first tree and try all possible rootings of the second tree to find an optimal solution.
Because of that, we can assume that we are given a rooted tree $T$, an unrooted tree $Q$
and we want to compute the edit distances between $T$ and every rooting of $Q$. 
To achieve that, we will compute the similarity matrix $S(T, Q)$. 
We will do this by computing recursively similarity matrices $S(T', Q)$ where $T'$ 
is some fragment of tree $T$. As the second argument of $S(T', Q)$ will always be $Q$ in our algorithm, we will use a shorthand $S(T') := S(T', Q)$.
Additionally, we call $S(T')$ the similarity matrix of tree $T'$.

\subsection{Dynamic programming scheme}

To compute $S(T')$, we consider the following cases:
\begin{enumerate}[(a)]
    \item \label{dp_a} $T'$ has no edges (it is a single vertex). Then the values in all valid cells are $0$.
    \item \label{dp_b} The root of $T'$ has only one child. Let $r$ be the root and $u$ be its only child. We choose if we want to match $ru$ to some edge $e$ of $Q[i,j)$ or not:
    \begin{equation*}
		S(T')_{i,j} = \max
			\begin{cases}
				S(T' - ru)_{i,j}\\
				\max_{e \in Q[i;j)} \set{ S(T' - ru)_{l_{i,j}(e)+1,r_{i,j}(e)} + \eta(ru, e) }
			\end{cases}
	\end{equation*}
    \item \label{dp_c} The root of $T'$ has more than one child. Let $R_{T'}$ be the subtree of the last child of the root of $T'$ along with the edge to the root of $T'$. Then:
    \begin{equation*}
        S(T')_{i,j} = \max_{i \leq k \leq j}\set { S(T' - R_{T'})_{i,k} + S(R_{T'})_{k,j}} 
    \end{equation*}
    In other words: $S(T') = S(T' - R_{T'}) \star S(R_{T'})$.
\end{enumerate}
The correctness of the first two cases is obvious. To see the correctness of the third case, let us first notice that our formula covers the case when all edges from $Q[i, j)$ are matched 
only to the edges of $T' - R_{T'}$. The same is true when they are matched only to the edges of $R_{T'}$. 
Otherwise, let us take any edge $e \in Q[i,j)$ that is matched to some edge in $T' - R_{T'}$ and any edge $f \in Q[i, j)$ that is matched to some edge in $R_{T'}$. From properties of the matching we get that both occurrences of $e$ in $Q(i),\ldots,Q(j-1)$ need to be before both occurrences of $f$ in $Q(i),\ldots,Q(j-1)$. From this we get that it is enough to consider each division of $Q[i, j)$ into two segments.

Now, let us estimate the time complexity of this approach. First, we notice that the number of recursive calls is linear: every tree $T'$ that we consider is either a prefix of children, with their subtrees, of some node or it is a subtree of some node with edge to its parent. Naive computation of cases (b) and (c) requires $\Oh(n^3)$ time. 
Thus, the algorithm works in $\Oh(n^4)$ time.

\subsection{Optimization to cubic} \label{sec:cubic}
To optimize our algorithm to $\Ot(n ^ 3)$ time, we will exploit special properties of similarity matrices. We will rely on the fact that $S(T')$ is a $2|E(T')|$-bounded similarity matrix.

Furthermore, we will use the fact that these matrices are almost row-monotone and column-monotone. The only exception to this property are invalid cells, however their value is always $-\infty$. 
Because of that, we are able to use the same computation model that Mao used to store row-monotone, column-monotone matrices with slight modifications to handle invalid cells. 

Let $l := 2n -2$. We consider the following operations for $2l \times 2l$ matrices:
\begin{itemize}
    \item
        create a new $2l \times 2l$ matrix $[-\infty]_{2l,2l}$,
    \item
        given a matrix $A$, create a copy of this matrix,
    \item
        given a matrix $A$, indexes $i', j'$ and a value $x$,
        create a new  $2l \times 2l$ matrix \\ $B := \rangemax(A, i', j', x)$, such that:
        \begin{equation*}
            B_{ij} =
                \begin{cases}
                    \max(A_{ij}, x) & \text{ if $i \leq i'$ and $j \geq j'$ and $j-i \leq l$} \\
                    A_{ij} & \text{ otherwise }
                \end{cases}
        \end{equation*}
\end{itemize}
Furthermore, given matrix $A$ we consider the following queries: 
\begin{itemize}
    \item get the value $A_{i,j}$,
    \item get the index $\mincol(A, i, x) = \min \{ j \mid A_{ij} \geq x \}$ or any index in $[1, 2l]$ if such an index does not exist,
    \item get the index $\maxrow(A, j, x) = \max \{ i \mid A_{ij} \geq x \}$ or any index in $[1, 2l]$ if such an index does not exist.
\end{itemize}
Note that the underlying data structure can keep a row-monotone, column-monotone matrix, and we can just modify queries to return $-\infty$ when reading invalid cells. This means we don't need to have the $j-i \leq l$ condition in the $\rangemax$ operation, thus we can use a simple 2D $\max$ operation (the same as Mao). All these operations can be performed with well-known data structures, such as persistent 2D segment trees in $\Ot(1)$ time.

From now on, we assume that we store similarity matrices using the above computation model. 
We will go through all three cases of our dynamic programming scheme and we will show how to efficiently solve them using the above data structure:

\subparagraph*{(\ref{dp_a}):}
We can easily initialize all valid cells to $0$ using $\Oh(n)$ $\rangemax$ queries. This gives us $\Ot(n^2)$ time for the entire algorithm.

\subparagraph*{(\ref{dp_b}):}
Let us recall the equation for this case:
\begin{equation*}
    S(T')_{i,j} = \max
        \begin{cases}
            S(T' - ru)_{i,j}\\
            \max_{e \in Q[i;j)} \set{ S(T' - ru)_{l_{i,j}(e)+1,r_{i,j}(e)} + \eta(ru, e) }
        \end{cases}
\end{equation*}
Naive computation of this equation relies on iterating through all cells we want to compute and for each $S(T')_{i, j}$ we iterate through each edge of $Q[i, j)$. Instead, we can do the opposite and iterate through each edge of $Q$ and update for each of them all relevant cells of $S(T')$. Thus, we first initialize $S(T')$ with $S(T'-ru)$ and then for each edge $e$ of $Q$ we make three $\rangemax$ operations. See Algorithm \ref{dp_b_opt} for exact values of the parameters.
\begin{algorithm}[!htb]
    \caption{Computation of case \ref{dp_b}}\label{dp_b_opt}
    \hspace*{\algorithmicindent} \textbf{Input} $T, Q, S(T'-ru)$\\
    \hspace*{\algorithmicindent} \textbf{Output} $S(T')$
    \begin{algorithmic}[1]
        \State $S(T') \gets S(T'-ru)$
        \ForAll{$e \in E(Q)$}
            \State $S(T') \gets \rangemax(S(T'), I(e)_1, I(e)_2 + 1, S(T' - ru)_{I(e)_{1}+1,I(e)_{2}} + \eta(ru, e) )$
            \State $S(T') \gets \rangemax(S(T'), I(e)_2, I(e)_3 + 1, S(T' - ru)_{I(e)_{2}+1,I(e)_{3}} + \eta(ru, e) )$
            \State $S(T') \gets \rangemax(S(T'), I(e)_3, I(e)_4 + 1, S(T' - ru)_{I(e)_{3}+1,I(e)_{4}} + \eta(ru, e) )$
        \EndFor
        \State \Return $S(T')$
    \end{algorithmic}
\end{algorithm}
A single computation of this case takes $\Ot(n)$ time, which gives us $\Ot(n^2)$ time for the entire algorithm.

\subparagraph*{(\ref{dp_c}):}
This case is a matrix multiplication $S(T') = S(T' - R_{T'}) \star S(R_{T'})$. To speed up the computation, we will relay on the fact that $S(T')$ is a $2|E(T')|$ -bounded similarity matrix.
The following lemma, which corresponds to Lemma 4.3 from Mao \cite{Mao}, shows that we can multiply similarity matrices in complexity dependent on the bounds of values of these matrices:

\begin{lemma}\label{mul_bounded}
    Let $A, B$ be $l \times l$ matrices. If $A$ is a $t_A$-bounded similarity matrix and $B$ is a $t_B$-bounded similarity matrix, then we can compute $C = A \star B$ in $\Ot(t_At_Bl)$ time.
\end{lemma}
Algorithm \ref{alg_mul_bounded} that computes $C = A\star B$ is the same as Algorithm 1 from Mao \cite{Mao}.
\begin{algorithm}[!htb]
    \caption{}\label{alg_mul_bounded}
    \begin{algorithmic}[1]
        \Function{MulBounded}{$A, B$}
            \State $C \gets [-\infty]_{l,l}$
            \ForAll{$j \in [1,l]$}
                \ForAll{$x \in [0,t_B]$}
                    \ForAll{$y \in [0,t_A]$}
                        \State $k \gets \maxrow(B, j, x)$
                        \State $i \gets \maxrow(A, k, y)$
                        \State $C \gets \rangemax(C, i, j, A_{i,k} + B_{k,j})$ \label{alg_mul_bounded_rangemax}
                    \EndFor
                \EndFor
            \EndFor
            \State \Return $C$
        \EndFunction
    \end{algorithmic}
\end{algorithm}

To prove the correctness of this algorithm, let us first notice that we update values of the resulting matrix only if all three pairs of indexes from line \ref{alg_mul_bounded_rangemax}: $(i, j)$, $(i, k)$, $(k, j)$ define valid cells. Otherwise, either one of $A_{i, k}$, $B_{k, j}$ is an invalid cell and we pass $-\infty$ as a value to the $\rangemax$ operation, or $(i, j)$ defines an invalid cell and the $\rangemax$ operation does not affect any cell.

Therefore, let us consider only $\rangemax$ operations for which all three pairs $(i, j)$, $(i, k)$, $(k, j)$ correspond to some valid cells. Let $C = A \star B$. Note that $A_{i, k} + B_{k, j} \leq C_{i, j}$ and due to the fact that row-monotone and column-monotone properties hold for all valid cells, we also have $A_{i, k} + B_{k, j} \leq C_{i', j'}$ for all valid cells $C_{i', j'}$ such that $i' \leq i$ and $j \leq j'$. 
Thus, in our resulting matrix we will not have any values that are too big.

It remains to prove that, for all resulting cells, the optimal values are achievable. Let us take any valid cell $C_{i', j'}$. Let $k^*$ be an index such that $C_{i', j'} = A_{i', k^*} + B_{k^*, j'}$. Now, let us consider the step of our algorithm when $j = j', x = B_{k^*, j'}$, and $y = A_{i', k^*}$. Due to the row-monotone, column-monotone properties and the fact that $C_{i', j'} = A_{i', k^*} + B_{k^*, j'}$ we first get that $k \geq k^*$ and $B_{k, j} = B_{k^*, j'}$. 
Then from the same argument we get that $i \geq i'$ and $A_{i, k} = A_{i', k^*}$. So, the $\rangemax$ operation will update $C_{i',j'}$ with $A_{i', k^*} + B_{k^*, j'}$. 

Thus, we proved that every time we multiply similarity matrices of trees $T' - R_{T'}$ and $R_{T'}$, we contribute $\Ot(|T' - R_{T'}||R_{T'}|n)$ time to the total time complexity of this case.  
To compute the sum of these components, we can correlate $|T' - R_{T'}||R_{T'}|$ with the number of pairs of edges $(e, f)$ where $e \in T' - R_{T'}$ and $f \in R_{T'}$.
Let us sum these pairs over all multiplications. It is easy to see that each pair of edges from tree $T'$ appears at most once in such a sum. Thus, the total time complexity of this case is equal to $\Ot(n^3)$.

\medskip
After considering all three cases, we can see that the whole algorithm works in $\Ot(n^3)$ time. Note here that if the input trees are of different sizes $n$ and $m$, then this algorithm works in $\Ot(nm^2)$ time.

%% file: paragraphs/subcubic.tex
\section{Subcubic algorithm for unrooted tree edit distance}
\label{sec:subcubic}

In this section, we show how to get a subcubic algorithm for the unrooted tree edit distance problem.
Our algorithm is based on the same decomposition scheme as Mao's algorithm.
However, we will start by showing a different bottom-up view at $\Ot(n^3)$ algorithm. 
This will give us a better intuition about Mao's decomposition scheme.

We can view the $\Ot(n^3)$ algorithm as a decomposition scheme that computes $S(T)$ using the following two types of transitions:
\begin{itemize}
    \item \textbf{Type I:} for trees $T_1, T_2$: 
    \begin{center}
        compute $S(T_1 + T_2)$ from $S(T_1)$ and $S(T_2)$
    \end{center}
    \item \textbf{Type II:} for trees $T_1, T_2$ such that $T_2$ is $T_1$ with one added edge going from the root up:
    \begin{center}
        compute $S(T_2)$ from $S(T_1)$
    \end{center}
\end{itemize}

Note that the total time complexity of type I transitions is $\Ot(n^3)$, while the total complexity of type II transitions is $\Ot(n^2)$. This gives us an idea that to get a subcubic complexity we can try to balance these transitions. To reduce the number of type I transitions, we can generalize transitions of type II. Instead of adding a single edge going up from the root of $T_1$, we can add a ``hat'', that is, a path going up from the root along with some additional subtrees. This is the general idea of Mao's decomposition scheme, which we will now formally describe. See Figure \ref{fig:czap} for an illustration of the ``hat'' structure.

First, we introduce additional definition that will help us define sub-problems occuring in the decomposition scheme. This definition corresponds to the synchronous subtree definition from Mao. 
\begin{definition}[Connected segment]
    Given an unrooted tree $T$, segment $T[i,j)$ is called a \textit{connected segment} if there is no edge that occurs in $T(i),\ldots,T(j-1)$ exactly once.        
\end{definition}
A connected segment can be alternatively defined by a selection of a vertex $v \in T$ and a connected interval of children of the vertex $v$ (we call $v$ the root of the segment), so that subtrees of these children along with the edges from $v$ to these children belong to this segment. Thus, a connected segment forms a connected subtree of tree $T$. 

Now let $\Delta$ be the block size -- a parameter, that we will set later. We decompose the computation of $S(T)$ into $\Oh(n/\Delta)$ transitions of two types:
\begin{itemize}
    \item \textbf{Type I:} for connected segments $T_1, T_2$ such that $|T_1| \geq \Delta$ and $|T_2| \geq \Delta$:
    \begin{center}
        compute $S(T_1 + T_2)$ from $S(T_1)$ and $S(T_2)$
    \end{center}
    \item \textbf{Type II:} for connected segments $T_1, T_2$ such that $T_1 \subset T_2$ and $|T_2| - |T_1| = \Oh(\Delta)$:
    \begin{center}
        compute $S(T_2)$ from $S(T_1)$
    \end{center}
\end{itemize}
To compute this decomposition, we will use the same algorithm as Mao did.

The following theorem, which corresponds to Theorem 4.5 from Mao \cite{Mao}, tells us how we can perform a type I transition efficiently. We prove this in Section \ref{sec:mat_mul}.

\begin{theorem}\label{th:fast_mul}
    Let $A$, $B$ be $n \times n$ similarity matrices of unrooted trees. 
    If $A$ is $t$-bounded, then $C = A \star B$ can be computed in $\mul(t, n) := \Ot(t^{0.8145}n^2)$ time.
\end{theorem}  

For type II transitions, in Section \ref{sec:type2} we prove the following theorem, which corresponds to Theorem 4.6 from Mao \cite{Mao}. 
\begin{theorem}\label{th:type2}
    Let $T_1, T_2$ be the connected segments of tree $T$ such that $T_1 \subset T_2$ and $|T_2| - |T_1| = \Oh(\Delta)$. Given $S(T_1)$ we can compute $S(T_2)$ in $\Ot(\mul(\Delta, n) + n\Delta^4)$ time.   
\end{theorem}
Note that in the above theorem we have slightly worse time complexity for the type II transition than $\Ot(\mul(\Delta, n) + n\Delta^3)$ time from Mao's algorithm.

\subsection{Implementation of the decomposition scheme}
\label{sec:decomp}
Now we show how to decompose the tree $T$ into $\Oh(n/\Delta)$ transitions of type I and II.
The decomposition procedure is recursive. If $L_T \neq R_T$, $|L_T| \geq \Delta$ and $|R_T| \geq \Delta$ then we use type I transition to compute $S(T)$ from $S(T - R_T)$ and $S(R_T)$. 

Otherwise, we use the type II transition. For that we need to find a connected segment $T'$ of tree $T$, such that $|T| - |T'| = \Oh(\Delta)$. We start with $T'$ equal to $T$ and we keep removing some part of it. If $\deg(\rot(T')) = 1$ then we remove the only edge that is adjacent to the root of $T'$. Otherwise, we remove the smaller tree among $L_{T'}$ and $R_{T'}$. We stop when the next removal would cause that $|T|-|T'| > 2\Delta$. Algorithm \ref{alg:main} contains the entire computation procedure of the similarity matrix $S(T)$. 

\begin{algorithm}[!htb]
    \caption{Computation of S(T)} \label{alg:main}
    \begin{algorithmic}[1]
        \Function{COMPUTE}{$T$}
            \If{$L_T \neq R_T$\  \textbf{and}\  $|L_T| \geq \Delta$ \textbf{and} $|R_T| \geq \Delta$}
                \State COMPUTE($T - R_T$)
                \State COMPUTE($R_T$)
                \State \textbf{Type I} transition: Compute $S(T) = S(T - R_T) \star S(R_T)$ using Theorem \ref{th:fast_mul}
            \Else
                \State $T' \gets T$
                \While {TRUE}
                    \State $T_{NEXT} \gets T'$
                    \If{$\deg(\rot(T')) = 1$}
                        \State  $T_{NEXT} \gets T' - uv$ \Comment{$uv$ is an edge adjacent to the root of $T'$}
                    \ElsIf{$\deg(\rot(T')) > 1$}
                        \State $T_{NEXT} \gets T' - min(L_{T'}, R_{T'})$
                    \EndIf
                    \If{$|T| - |T_{NEXT}| > 2\Delta$ \textbf{or} $T' = \emptyset$}
                        \State \textbf{break}
                    \EndIf 
                    \State $T' \gets T_{NEXT}$
                \EndWhile
                \State COMPUTE($T'$)
                \State \textbf{Type II} transition: Compute $S(T)$ from $S(T')$ using Theorem \ref{th:type2}
            \EndIf
            \State \Return $S(T)$
        \EndFunction
    \end{algorithmic}
\end{algorithm}

Now we prove that Algorithm \ref{alg:main} indeed decomposes the computation of $S(T)$ into $\Oh(n/\Delta)$ transitions of type I and II. For transitions of type I, we can notice that they correspond to merging of two trees of size at least $\Delta$, and therefore there can be at most $\Oh(n/\Delta)$ such transitions.

For transitions of type II, we can divide them into three groups:
\begin{itemize}
    \item Large transitions: $T' \neq \emptyset$ and $|T| - |T'| \geq \Delta$
    \item Small transitions: $T' \neq \emptyset$ and $|T| - |T'| < \Delta$
    \item Base transitions: $T' = \emptyset$
\end{itemize}
First, we can notice that trees $T - T'$ are disjoint between different transitions of type II. We can therefore bound the number of large transitions of type II by $\Oh(n/\Delta)$. 

Now, let us notice that after each small transition there will be a transition of type I. Furthermore, no two small transitions are followed by the same type I transition. Thus, we can bound the number of small transitions by the number of transitions of type I, which we have already shown is $\Oh(n/\Delta)$.

Finally, for base transitions we can notice that they correspond to the leaves in the recursion tree. The number of leaves in this recursion tree is exactly greater by $1$ than the number of binary branches, i.e.~type I transitions. Thus, we showed that the total number of transitions of type I and II in our decomposition scheme is bounded by $\Oh(n/\Delta)$.

\input{paragraphs/matrix_mul.tex}

\input{paragraphs/type2.tex}

\subsection{Total running time}
\label{sec:time}
Now we analyze the total running time of our algorithm. Unlike Mao, we also present the total running time in terms of matrix multiplication complexity. 

We start with an analysis in terms of matrix multiplication complexity for similarity matrices. Let us assume that we can multiply similarity matrices in $\Oh(t^a n^b)$ time when one of the matrices is $t$-bounded and $a \in [0, 1]$.

For type I transitions, we first assume that we use a slower $\Oh(tn^b)$ time algorithm for matrix multiplication. Now, we use the small-to-large technique. 
Each multiplication of the similarity matrices of two trees $T_1$, $T_2$ takes $\Oh(\min(|T_1|,|T_2|)n^b)$ time. We can therefore look at it as if each node from the smaller tree adds $\Oh(n^b)$ time to our complexity. Through our algorithm any node can be at most $\Oh(\log n)$ times in a smaller tree, so all type I transitions take $\Ot(n^{1+b})$ time when we use an $\Oh(t n^b)$ time algorithm for matrix multiplication.

Now, let us replace the $\Oh(t n^b)$ time algorithm with the $\Oh(t^a n^b)$ time algorithm. 
As both trees in type I transitions have size at least $\Delta$, we save (multiplicatively) $\Omega(\Delta^{1-a})$ time on each multiplication. Thus, the total running time of type I transitions is $\Ot(n^{1+b} / \Delta ^{1-a})$ time.

For type II transitions, we can compute their total running time directly:
\begin{gather*}
    \Ot((n/\Delta)(\Delta^a n^b + n\Delta^4)) = \Ot(n^{1+b}/\Delta^{1-a} + n^2\Delta^3)
\end{gather*}

Thus, the total running time of our algorithm is $\Ot(n^{1+b}/\Delta^{1-a} + n^2\Delta^3)$. It reaches the optimum when $\Delta = n^{(b-1)/(4-a)}$ which gives us that the whole algorithm works in $\Ot(n^{(3b-2a+5)/(4-a)})$ time.

Now assume that we can compute the $(\min, +)$ product of two bounded-difference matrices in $\Ot(n^c)$ time. From the proof of Theorem \ref{th:fast_mul} we know that we can multiply similarity matrices in $\Oh(t^{(2c-4)/(c-1)} n^2)$ time if one of the matrices is $t$-bounded. Thus, we can set $a := (2c-4)/(c-1)$ and $b := 2$. This gives us that the total running time of our algorithm is $\Ot(n^{(7c-3)/2c})$ time.

From Theorem \ref{th:bound_mat_mul} we know that $c \leq (3 + \omega)/2$ where $\omega$ is the fast matrix multiplication exponent. Therefore, we get that our algorithm works in $\Ot(n^{(7\omega + 15)/(2\omega + 6)})$ time.

Finally, as $\omega < 2.37286$ \cite{omega} we get that we can compute the tree edit distance between two unrooted trees of size $n$ in $\Ot(n^{2,9417})$ time.

\subparagraph*{Trees of different sizes:}
Now we address the case when input trees have different sizes. Let $m = |T|$ and $n = |Q|$. 
In such a case, our decomposition scheme decomposes the computation of $S(T)$ into $\Oh(m/\Delta)$ transitions. For type I transitions, we can see that the complexity changes to $\Oh(mn^b /\Delta^{1-a} )$ time. For type II transitions we have:
\begin{gather*}
    \Ot((m/\Delta)(\Delta^a n^b + n\Delta^4)) = \Ot(mn^{b}/\Delta^{1-a} + mn\Delta^3)
\end{gather*}
Now it is easy to see that in all the above complexities, one factor $\Oh(n)$ changes to $\Oh(m)$. Thus, we get the proof of Theorem \ref{th:main}.

%% file: paragraphs/matrix_mul.tex
\subsection{Type I transitions}
\label{sec:mat_mul}
In this section, we show how to efficiently multiply similarity matrices of unrooted trees and prove Theorem \ref{th:fast_mul}.
Mao showed how to multiply similarity matrices of rooted trees in $\Ot(t^{0.9038}n^2)$ time, where 
$t$ is a bound on the entries of one of the matrices. For unrooted trees, we will use his algorithm with a slight modification.

Mao's algorithm is a recursive procedure that uses the $(\max, +)$ matrix multiplication of bounded-difference matrices as a sub-procedure. 
Lately, Chi et al. \cite{Matrix} showed a better algorithm for the $(\max, +)$ product of bounded-difference matrices. 
This algorithm allows us to get a better exponent in Mao's matrix multiplication. 

\begin{theorem}[{\cite{Matrix}}]\label{th:bound_mat_mul}
    There is an $\Ot(n^{(3 + \omega)/2}) \leq \Ot(n^{2.687})$ time randomized algorithm that computes the 
    $(\min, +)$ product of any two $n \times n$ bounded-difference matrices.
\end{theorem}

To analyze the improvement in the exponent, we will use the following lemma implicitly proven by Mao:

\begin{lemma}[{\cite[Section 4.4]{Mao}}]
\label{Mao-mul}
    Assume there is an $\Ot(n^c)$ time algorithm that computes the $(\min, +)$ product of any two 
    $n \times n$ bounded-difference matrices. Let $A$, $B$ be row-monotone, column-monotone and 
    finite-upper-triangular-bounded-difference
    $n \times n$ matrices whose entries on the main diagonals are zero.     
    If $A$ is $t$-bounded-upper-triangular and $\delta$ is a positive value, then $C = A \star B$ 
    can be computed in $\Ot(n^2t^2/\delta + n^2 \delta^{c-2})$ time.
\end{lemma}

The above complexity reaches the optimum when $n^2t^2/\delta = n^2 \delta^{c-2}$. 
By setting $\delta = t^{2/(c-1)}$ we get that we can multiply the similarity matrices of rooted trees in $\Ot(t^{(2c-4)/(c-1)}n^2)$ time, 
if one of the matrices has values bounded by $t$.

Now we want to use this algorithm also for unrooted trees. However, we need to make sure 
that the values from invalid cells do not affect the values of valid cells. To easily handle this,  
before calling the algorithm from Lemma \ref{Mao-mul} we fix the values of invalid cells based on the values of valid cells.

Let $A$, $B$ be similarity matrices of unrooted trees. We want to compute $C = A \star B$.
For every invalid cell $A_{i, j}$ we set new value of $A_{i, j}$ as the maximum of valid cells $A_{i',j'}$, 
such that $i \leq i' \leq j' \leq j$. Let $A'$ be the resulting matrix. 
We can easily compute this transformation in $\Oh(n^2)$ time by using simple recursive equation 
$A'_{i,j} = \max(A'_{i+1, j}, A'_{i, j-1})$. Similarly, we convert matrix $B$ into matrix $B'$. 

The obtained matrices $A'$ and $B'$ are obviously monotone and finite-upper-triangular. 
To prove that these new matrices are bounded-difference let us use induction on $i + j$.
For a valid cell, our inductive thesis is satisfied. Now, take any invalid cell $A'_{i, j}$. 
Then, by inductive hypothesis we have that $A'_{i-1, j} \in [A'_{i-1, j-1}, A'_{i-1, j-1} + 2]$ 
and $A'_{i, j-1} \in [A'_{i-1, j-1}, A'_{i-1, j-1} + 2]$.
Thus, $|A'_{i,j} - A'_{i-1, j}| \leq 2$ and $|A'_{i,j} - A'_{i, j-1}| \leq 2$, so 
matrices $A'$ and $B'$ are bounded-difference matrices.

Let $C' = A' \star B'$. It remains to show that $C' =_{valid} C$. Let us take any valid cell $C'_{i, j}$.
There is a position $k$ such that $C'_{i, j} = A'_{i, k} + B'_{k, j}$. 
If one of $A'_{i, k}$, $B'_{k, j}$ is an invalid cell, then the other would be equal to $-\infty$
and we would have $C'_{i, j} < 0$. But $C'_{i, j} \geq A'_{i, i} + B'_{i, j} \geq 0$, 
which gives us a contradiction.
Thus, invalid cells of $A'$ and $B'$ have no effect on valid cells of $C'$, so $C' =_{valid} C$.    

This gives us that we can compute $C = A \star B$ in $\Ot(t^{(2c-4)/(c-1)}n^2)$ time 
assuming we can compute the $(\min, +)$  product of any two 
$n \times n$ bounded-difference matrices in $\Ot(n^c)$ time. 

Combining this with Theorem~\ref{th:bound_mat_mul} we get the proof of Theorem \ref{th:fast_mul}:

%% file: paragraphs/type2.tex
\subsection{Type II transitions}
\label{sec:type2}

In this section, we prove Theorem \ref{th:type2}. Let $T_1$, $T_2$ be connected segments of tree $T$ such that $T_1 \subset T_2$ and $|T_2| - |T_1| = \Oh(\Delta)$. We want to compute $S(T_2)$ given $S(T_1)$. 

Let us consider the path $p = e_1e_2\ldots e_k$ in tree $T$ that goes from the root of $T_2$ to the root of $T_1$. For $1\leq i \leq k$, let $L_i$ be the subtree of tree $T_2$ consisting of siblings of edge $e_i$ that are to the left of $e_i$ and all descendants of these edges. Additionally, let $L_{k+1}$ be a subtree of tree $T_2$ consisting of descendants of edge $e_k$ that are to the left of tree $T_1$. Similarly, we denote subtrees that are to the right of path $p$ by $R_i$ for $1\leq i \leq k+1$. See Figure \ref{fig:czap}.

\begin{figure}[!h]
    \begin{center}
    \includegraphics[scale=1.15]{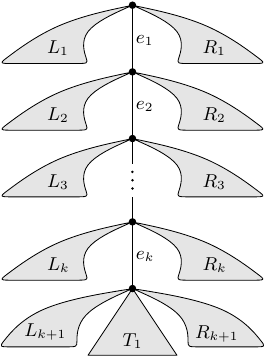}
    \caption{Tree $T_2$. In the type II transition we compute $S(T_2)$ based on $S(T_1)$.}\label{fig:czap}
    \end{center}
\end{figure}

For $1\leq i \leq j \leq k+1$, by $L_{i, j}$ we denote the tree $L_i + L_{i+1} + \ldots + L_{j}$ and by $R_{i, j}$ we denote the tree $R_{j} + R_{j-1} + \ldots + R_{i}$.

First, for all trees $L_{i,j}$, $R_{i, j}$, we compute the similarity matrices $S(L_{i,j})$, $S(R_{i, j})$ using the $\Ot(nm^2)$ time algorithm from Section \ref{sec:cubic}. We notice that in one run of this algorithm, we can compute $S(L_{i,j})$ for a fixed $i$ and all $j \geq i$. Thus, we only need to use that algorithm $\Oh(\Delta)$ times, which in total gives us $\Ot(n \Delta ^3)$ time for that step.

Now, we consider two cases:
\begin{enumerate}[(a)]
    \item None of the edges $e_1,\ldots,e_k$ is matched to some edge of tree $Q$.
    \item At least one of the edges $e_1,\ldots,e_k$ is matched to some edge of tree $Q$.
\end{enumerate}

For case (a), we can use Theorem \ref{th:fast_mul} to compute $S(L_{1, k+1}) \star S(T_1) \star S(R_{1, k+1})$ in $\mul(\Delta, n)$ time. For case (b), we define a restricted version of the similarity matrix for trees $T'$, such that the root of $T'$ has only one child.

\begin{definition}[Restricted similarity matrix]
    Let $T'$ be a rooted tree such that the root $u$ of $T'$ has only one child $v$.
    The \textit{restricted similarity matrix}  $\hat{S}(T')$ is a $4|E(Q)| \times 4|E(Q)|$ matrix where:
    \begin{gather*}
        \hat{S}(T')_{i,j} =
            \begin{cases}
                \max_{e \in Q[i, j)}\{ \simi(T' - uv, \sub(e) - e) + \eta(uv, e) \} & \text{ if } i \leq j, j - i \leq 2|E(Q)| \\
                -\infty & \text{ otherwise }
            \end{cases}
    \end{gather*}
\end{definition}

We will compute $\hat{S}(\sub(e_i))$ for all edges $e_i$ on path $p$. But first, let us assume that we have already computed these restricted similarity matrices. We show how to compute $S(T_2)$ using these matrices.

As mentioned before we can first initialize $S(T_2)$ with $S(L_{1, k+1}) \star S(T_1) \star S(R_{1, k+1})$. Now, let us consider a single valid cell $S(T_2)_{i, j}$ that we want to compute. To cover case (b), we can iterate through:
\begin{itemize}
    \item First edge $e_x$ from path $p$ that is matched.
    \item Edge $e \in Q[i, j)$ such that $e_x$ is matched to $e$.
\end{itemize}
Then, as shown in  Figure \ref{fig:top}, we have:
\begin{itemize}
    \item Edges from $L_{1, x}$ are matched to edges from $Q[i, l_{i,j}(e))$ contributing $S(L_{1,x})_{i, l_{i,j}(e)}$.
    \item Edges from $\sub(e_x)$ are matched to edges from $Q[l_{i,j}(e), r_{i,j}(e)+1)$ contributing\\ $\hat{S}(\sub(e_x))_{l_{i,j}(e), r_{i,j}(e)+1}$.
    \item Edges from $R_{1, x}$ are matched to edges from $Q[r_{i,j}(e)+1, j)$ contributing $S(R_{1,x})_{r_{i,j}(e)+1, j}$.
\end{itemize}
\begin{figure}[!h]
    \begin{center}
    \includegraphics[scale=1.15]{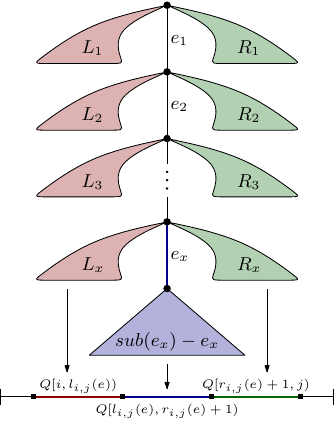}
    \caption{Computing $S(T_2)$ from $\hat{S}(\sub(e_x))$.}\label{fig:top}
    \end{center}
\end{figure}
Note that we allow edge $e_x$ to be matched to some other edge than $e$ from $Q[l_{i,j}(e), r_{i,j}(e)+1)$. However, the cost of such matching will be at least as good as the cost of best matching that matches $e_x$ to $e$, thus it does not affect the correctness of our algorithm.

This gives us that we can compute $S(T_2)$ from all $\hat{S}(\sub(e_i))$ in $\Ot(n^3 \Delta)$. To optimize this, we can use a similar idea to the one from Algorithm \ref{dp_b_opt}. Instead of calculating all the cells from $S(T_2)$ separately, we can do it globally. We iterate through:
\begin{itemize}
    \item First edge $e_x$ from path $p$ that is matched.
    \item Edge $e \in Q$ such that $e_x$ is matched to $e$.
    \item Occurrence $(I(e)_c, I(e)_{c+1})$ of the edge $e$ in the doubled Euler tour.
    \item Cost $a$ of matching edges from $L_{1, x}$.
    \item Cost $b$ of matching edges from $R_{1, x}$.
\end{itemize}
Now we can find the shortest segment $[i, j)$, such that it contains $(I(e)_c, I(e)_{c+1})$, the cost of matching $L_{1,x}$ to $Q[i, I(e)_c)$ is at least $a$ and the cost of matching $R_{1, x}$ to $Q[I(e)_{c+1} + 1, j)$ is at least $b$. Then we can use a single $\rangemax$ operation to update all valid cells $S(T_2)_{i',j'}$ for which $i' \leq i$ and $j \leq j'$. See Algorithm \ref{alg:top} for details. 

This gives us that we can compute $S(T_2)$ from $\hat{S}(\sub(e_i))$'s in $\Ot(n \Delta^3)$ time.
\begin{algorithm}[!htb]
    \caption{Computation of $S(T_2)$ from $S(T_1)$ and $\hat{S}(\sub(e_i))$'s} \label{alg:top}
    \begin{algorithmic}[1]
        \State $S(T_2) \gets S(L_{1, k+1}) \star S(T_1) \star S(R_{1, k+1})$
        \ForAll{$x \in [1, k]$}
            \ForAll{$e \in E(Q)$ \textbf{and} $c \in [1, 3]$} 
                \ForAll{$a \in [0, 2(|T_2| - |T_1|)]$}
                \ForAll{$b \in [0, 2(|T_2| - |T_1|)]$}
                    \State $i \gets \maxrow(S(L_{1, x}), I(e)_c, a)$
                    \State $j \gets \mincol(S(R_{1, x}), I(e)_{c+1} + 1, b)$
                    \State $val \gets S(L_{1,x})_{i, I(e)_c} + \hat{S}(\sub(e_x))_{I(e)_c, I(e)_{c+1} + 1} + S(R_{1,x})_{I(e)_{c+1} + 1, j}$
                    \State $S(T_2) \gets \rangemax(S(T_2), i, j, val)$
                \EndFor
                \EndFor
            \EndFor
        \EndFor
        \State \Return $S(T_2)$
    \end{algorithmic}
\end{algorithm}

Now, it remains to show how to compute all $\hat{S}(\sub(e_i))$'s. We compute them starting with $e_k$ and going up to $e_1$. Thus, assume that we have already computed $\hat{S}(\sub(e_i))$'s for all $i$ greater than some $x$. We want to compute $\hat{S}(\sub(e_x))$. Let us consider two cases:
\begin{enumerate}[(a)]
    \item None of the edges $e_{x+1}, \ldots, e_k$ is matched.
    \item At least one of the edges $e_{x+1}, \ldots, e_k$ is matched.
\end{enumerate}

\subparagraph*{Case (a):} Let us consider the valid cell $S(T_2)_{i, j}$ that we want to compute. To do this we iterate through:
\begin{itemize}
    \item Edge $e \in Q[i, j)$ such that $e_x$ is matched to $e$.
    \item Values $i', j'$ such that $l_{i,j}(e) \leq i' \leq j' \leq r_{i, j}(e)$ and edges from $T_1$ are matched to edges from $Q[i', j')$.
\end{itemize}
Then, as shown in Figure \ref{fig:bottom}, we have:
\begin{itemize}
    \item Edge $e_x$ is matched to edge $e$ contributing $\eta(e_x, e)$.
    \item Edges from $L_{x+1, k+1}$ are matched to edges from $Q[l_{i,j}(e), i')$ contributing \\$S(L_{x+1, k+1})_{l_{i,j}(e), i'}$.
    \item Edges from $T_1$ are matched to edges from $Q[i', j')$ contributing $S(T_1)_{i',j'}$.
    \item Edges from $R_{x+1, k+1}$ are matched to edges from $Q[j', r_{i,j}(e))$ contributing \\$S(R_{x+1,k+1})_{j', r_{i,j}(e)}$.
\end{itemize}
It is easy to see that this contribution applies not only to $S(T_2)_{i, j}$, but also to all valid cells $S(T_2)_{i'', j''}$ such that $i'' \leq l_{i,j}(e)$ and $r_{i, j}(e) + 1 \leq j''$. Thus, we do not need to compute cells $S(T_2)_{i, j}$ separately. This gives us $\Ot(n^3)$ time algorithm for case (a). We can optimize it to $\Ot(n \Delta^2)$ time in a similar way as we optimized Algorithm \ref{alg:top}.
See Algorithm \ref{alg:restricted} for full details.

\subparagraph*{Case (b):} Let $S(T_2)_{i, j}$ be some valid cell we want to compute. To do this, we iterate through:
\begin{itemize}
    \item Edge $e \in Q[i, j)$ such that $e_x$ is matched to $e$.
    \item First edge $e_y$ from $e_{x+1},\ldots,e_{k+1}$ that is matched.
    \item Values $i', j'$ such that $l_{i,j}(e) \leq i' \leq j' \leq r_{i, j}(e)$ and edges from $\sub(e_y)$ are matched to edges from $Q[i', j')$.
\end{itemize}
Then, as shown in Figure \ref{fig:middle}, we have:
\begin{itemize}
    \item Edge $e_x$ is mapped to edge $e$ contributing $\eta(e_x, e)$.
    \item Edges from $L_{x+1, y}$ are matched to edges from $Q[l_{i,j}(e), i')$ contributing $S(L_{x+1, y})_{l_{i,j}(e), i'}$.
    \item Edges from $\sub(e_y)$ are matched to edges from $Q[i', j')$ contributing $\hat{S}(\sub(e_y))_{i',j'}$.
    \item Edges from $R_{x+1, y}$ are matched to edges from $Q[j', r_{i,j}(e))$ contributing $S(R_{x+1,y})_{j', r_{i,j}(e)}$.
\end{itemize}
Similar to the case (a), we can optimize this case to $\Ot(n \Delta^3)$ time. See Algorithm \ref{alg:restricted} for details.

\begin{figure}[!htb]
    \centering
    \begin{subfigure}{0.45\textwidth}
        \centering
        \includegraphics[scale=1.1]{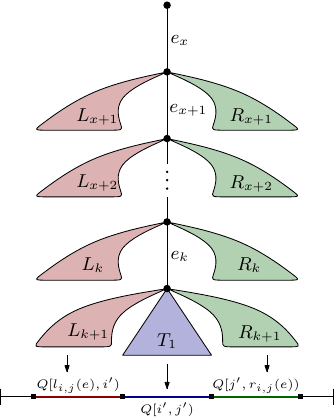}
        \caption{Computing $\hat{S}(\sub(e_x))$ from $S(T_1)$.}\label{fig:bottom}
    \end{subfigure}
    \hfill
    \begin{subfigure}{0.45\textwidth}
        \centering
        \includegraphics[scale=1.1]{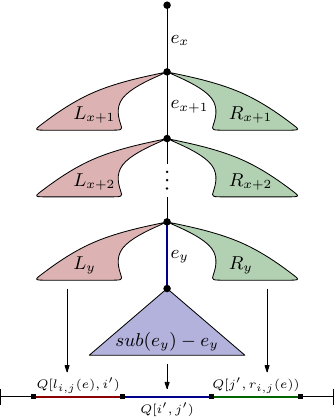}
        \caption{Computing $\hat{S}(\sub(e_x))$ from $\hat{S}(\sub(e_y))$.}\label{fig:middle}
    \end{subfigure}
    \caption{Computing $\hat{S}(\sub(e_x))$}
\end{figure}



\begin{algorithm}[!htb]
    \caption{Computation of $\hat{S}(\sub(e_x))$ from $\hat{S}(\sub(e_i))$ for $i > y$} \label{alg:restricted}
    \begin{algorithmic}[1]
        \State $\hat{S}(\sub(e_x)) \gets [-\infty]_{4|E(Q)|, 4|E(Q)|}$
        \ForAll{$e \in E(Q)$ \textbf{and} $c \in [1, 3]$} \Comment{Case (a)}
            \ForAll{$a \in [0, 2(|T_2| - |T_1|)]$}
            \ForAll{$b \in [0, 2(|T_2| - |T_1|)]$}
                \State $i \gets \mincol(S(L_{x+1, k+1}), I(e)_c, a)$
                \State $j \gets \maxrow(S(R_{x+1, k+1}), I(e)_{c+1}, b)$
                \State $val \gets \eta(e_x, e) + S(L_{x+1, k+1})_{I(e)_c, i} + S(T_1)_{i,j} + S(R_{x+1,k+1})_{j, I(e)_{c+1}}$
                \State $\hat{S}(\sub(e_x)) \gets \rangemax(S(T_2), I(e)_c, I(e)_{c+1} + 1, val)$
            \EndFor
            \EndFor
        \EndFor
        \ForAll{$y \in [x+1, k]$} \Comment{Case (b)}
            \ForAll{$e \in E(Q)$ \textbf{and} $c \in [1, 3]$} 
                \ForAll{$a \in [0, 2(|T_2| - |T_1|)]$}
                \ForAll{$b \in [0, 2(|T_2| - |T_1|)]$}
                    \State $i \gets \mincol(S(L_{x+1, y}), I(e)_c, a)$
                    \State $j \gets \maxrow(S(R_{x+1, y}), I(e)_{c+1}, b)$
                    \State $val \gets \eta(e_x, e) + S(L_{x+1, y})_{I(e)_c, i} + \hat{S}(\sub(e_y))_{i,j} + S(R_{x+1,y})_{j, I(e)_{c+1}}$
                    \State $\hat{S}(\sub(e_x)) \gets \rangemax(S(T_2), I(e)_c, I(e)_{c+1} + 1, val)$
                \EndFor
                \EndFor
            \EndFor
        \EndFor
        \State \Return $\hat{S}(\sub(e_x))$
    \end{algorithmic}
\end{algorithm}

Thus, we can compute single $\hat{S}(\sub(e_x))$ in $\Ot(n \Delta^3)$ time and all of them in $\Ot(n \Delta^4)$ time. Therefore, a whole single transition of type II can be computed in $\Ot(\mul(\Delta, n) + n \Delta^4)$ time, which finishes the proof of Theorem \ref{th:type2}.

%% file: paragraphs/remarks.tex
\section{Final remarks}
We have presented the first truly subcubic algorithm for the unweighted variant of the unrooted tree edit distance. However, as mentioned before our algorithm has a slightly worse exponent than the best-known algorithm for rooted trees. The difference comes from Theorem \ref{th:type2} where we have time complexity $\Ot(\mul(\Delta, n) + n\Delta^4)$, while Mao has $\Ot(\mul(\Delta, n) + n\Delta^3)$ in corresponding theorem. It will be interesting to see if our algorithm could be optimized to match the time complexity of Mao's algorithm.

For the weighted tree edit distance probably the most interesting open problem is a question whether there exists a weakly subcubic algorithm for that version. One of the lower bounds \cite{bound} that claims this problem cannot be solved in a truly subcubic time is based on APSP conjecture. However, for APSP weakly subcubic time algorithms are already known. Williams \cite{apsp} showed that the APSP problem can be solved in $\Oh(n^3 / 2^{\Omega(\sqrt{\log n})})$ time. Thus, we can try to obtain a weakly subcubic algorithm for the weighted tree edit distance by showing a reduction to the APSP problem.

%% file: main.bbl
\begin{thebibliography}{10}

\bibitem{omega}
Josh Alman and Virginia~Vassilevska Williams.
\newblock A refined laser method and faster matrix multiplication.
\newblock In D{\'{a}}niel Marx, editor, {\em Proceedings of the 2021 {ACM-SIAM}
  Symposium on Discrete Algorithms, {SODA} 2021, Virtual Conference, January 10
  - 13, 2021}, pages 522--539. {SIAM}, 2021.
\newblock \href {https://doi.org/10.1137/1.9781611976465.32}
  {\path{doi:10.1137/1.9781611976465.32}}.

\bibitem{IMG1}
J.~Bellando and R.~Kothari.
\newblock Region-based modeling and tree edit distance as a basis for gesture
  recognition.
\newblock In {\em Proceedings 10th International Conference on Image Analysis
  and Processing}, pages 698--703, 1999.
\newblock \href {https://doi.org/10.1109/ICIAP.1999.797676}
  {\path{doi:10.1109/ICIAP.1999.797676}}.

\bibitem{bound}
Karl Bringmann, Pawe\l{} Gawrychowski, Shay Mozes, and Oren Weimann.
\newblock Tree edit distance cannot be computed in strongly subcubic time
  (unless apsp can).
\newblock {\em ACM Trans. Algorithms}, 16(4), jul 2020.
\newblock \href {https://doi.org/10.1145/3381878} {\path{doi:10.1145/3381878}}.

\bibitem{XML1}
Peter Buneman, Martin Grohe, and Christoph Koch.
\newblock Path queries on compressed xml.
\newblock In {\em Proceedings of the 29th International Conference on Very
  Large Data Bases}, pages 141--152. Morgan Kaufmann, 2003.
\newblock \href {https://doi.org/10.1016/B978-012722442-8/50021-5}
  {\path{doi:10.1016/B978-012722442-8/50021-5}}.

\bibitem{XML2}
Sudarshan~S. Chawathe.
\newblock Comparing hierarchical data in external memory.
\newblock In {\em Proceedings of the 25th International Conference on Very
  Large Data Bases}, VLDB '99, pages 90--101, San Francisco, CA, USA, 1999.
  Morgan Kaufmann Publishers Inc.
\newblock \href {https://doi.org/10.5555/645925.671669}
  {\path{doi:10.5555/645925.671669}}.

\bibitem{Chen}
Weimin Chen.
\newblock New algorithm for ordered tree-to-tree correction problem.
\newblock {\em J. Algorithms}, 40(2):135--158, 2001.
\newblock \href {https://doi.org/10.1006/jagm.2001.1170}
  {\path{doi:10.1006/jagm.2001.1170}}.

\bibitem{Matrix}
Shucheng Chi, Ran Duan, Tianle Xie, and Tianyi Zhang.
\newblock Faster min-plus product for monotone instances.
\newblock In {\em Proceedings of the 54th Annual ACM SIGACT Symposium on Theory
  of Computing}, STOC 2022, pages 1529--1542, New York, NY, USA, 2022.
  Association for Computing Machinery.
\newblock \href {https://doi.org/10.1145/3519935.3520057}
  {\path{doi:10.1145/3519935.3520057}}.

\bibitem{n3}
Erik~D. Demaine, Shay Mozes, Benjamin Rossman, and Oren Weimann.
\newblock An optimal decomposition algorithm for tree edit distance.
\newblock {\em ACM Trans. Algorithms}, 6(1), dec 2010.
\newblock \href {https://doi.org/10.1145/1644015.1644017}
  {\path{doi:10.1145/1644015.1644017}}.

\bibitem{Dudek}
Bartlomiej Dudek and Pawel Gawrychowski.
\newblock {Edit Distance between Unrooted Trees in Cubic Time}.
\newblock In Ioannis Chatzigiannakis, Christos Kaklamanis, D{\'a}niel Marx, and
  Donald Sannella, editors, {\em 45th International Colloquium on Automata,
  Languages, and Programming (ICALP 2018)}, volume 107 of {\em Leibniz
  International Proceedings in Informatics (LIPIcs)}, pages 45:1--45:14,
  Dagstuhl, Germany, 2018. Schloss Dagstuhl--Leibniz-Zentrum fuer Informatik.
\newblock URL: \url{http://drops.dagstuhl.de/opus/volltexte/2018/9049}, \href
  {https://doi.org/10.4230/LIPIcs.ICALP.2018.45}
  {\path{doi:10.4230/LIPIcs.ICALP.2018.45}}.

\bibitem{Durr}
Anita D{\"{u}}rr.
\newblock Improved bounds for rectangular monotone min-plus product and
  applications.
\newblock {\em Inf. Process. Lett.}, 181:106358, 2023.
\newblock \href {https://doi.org/10.1016/j.ipl.2023.106358}
  {\path{doi:10.1016/j.ipl.2023.106358}}.

\bibitem{XML3}
Paolo Ferragina, Fabrizio Luccio, Giovanni Manzini, and Senthilmurugan
  Muthukrishnan.
\newblock Compressing and indexing labeled trees, with applications.
\newblock {\em J. ACM}, 57, 11 2009.
\newblock \href {https://doi.org/10.1145/1613676.1613680}
  {\path{doi:10.1145/1613676.1613680}}.

\bibitem{BIO2}
Matthias H{\"o}chsmann, Thomas T{\"o}ller, Robert Giegerich, and Stefan Kurtz.
\newblock Local similarity in rna secondary structures.
\newblock In {\em Computational Systems Bioinformatics. CSB2003. Proceedings of
  the 2003 IEEE Bioinformatics Conference. CSB2003}, pages 159--168, 2003.
\newblock \href {https://doi.org/10.1109/CSB.2003.1227315}
  {\path{doi:10.1109/CSB.2003.1227315}}.

\bibitem{IMG2}
Philip Klein, Srikanta Tirthapura, Daniel Sharvit, and Ben Kimia.
\newblock A tree-edit-distance algorithm for comparing simple, closed shapes.
\newblock In {\em Proceedings of the Eleventh Annual ACM-SIAM Symposium on
  Discrete Algorithms}, SODA '00, pages 696--704, USA, 2000. Society for
  Industrial and Applied Mathematics.
\newblock \href {https://doi.org/10.5555/338219.338628}
  {\path{doi:10.5555/338219.338628}}.

\bibitem{Klein}
Philip~N. Klein.
\newblock Computing the edit-distance between unrooted ordered trees.
\newblock In Gianfranco Bilardi, Giuseppe~F. Italiano, Andrea Pietracaprina,
  and Geppino Pucci, editors, {\em Algorithms - {ESA} '98, 6th Annual European
  Symposium, Venice, Italy, August 24-26, 1998, Proceedings}, volume 1461 of
  {\em Lecture Notes in Computer Science}, pages 91--102. Springer, 1998.
\newblock \href {https://doi.org/10.1007/3-540-68530-8\_8}
  {\path{doi:10.1007/3-540-68530-8\_8}}.

\bibitem{IMG3}
Philip~N. Klein, Thomas~B. Sebastian, and Benjamin~B. Kimia.
\newblock Shape matching using edit-distance: An implementation.
\newblock In {\em Proceedings of the Twelfth Annual ACM-SIAM Symposium on
  Discrete Algorithms}, SODA '01, pages 781--790, USA, 2001. Society for
  Industrial and Applied Mathematics.
\newblock \href {https://doi.org/10.5555/365411.365779}
  {\path{doi:10.5555/365411.365779}}.

\bibitem{Mao}
Xiao Mao.
\newblock Breaking the cubic barrier for (unweighted) tree edit distance.
\newblock In {\em 62nd {IEEE} Annual Symposium on Foundations of Computer
  Science, {FOCS} 2021, Denver, CO, USA, February 7-10, 2022}, pages 792--803.
  {IEEE}, 2021.
\newblock \href {https://doi.org/10.1109/FOCS52979.2021.00082}
  {\path{doi:10.1109/FOCS52979.2021.00082}}.

\bibitem{IMG4}
Thomas Sebastian, Philip Klein, and Benjamin Kimia.
\newblock Recognition of shapes by editing shock graphs.
\newblock {\em Pattern Analysis and Machine Intelligence, IEEE Transactions
  on}, 26:550 -- 571, 06 2004.
\newblock \href {https://doi.org/10.1109/TPAMI.2004.1273924}
  {\path{doi:10.1109/TPAMI.2004.1273924}}.

\bibitem{RNA}
Bruce~A. Shapiro and Kaizhong Zhang.
\newblock {Comparing multiple RNA secondary structures using tree comparisons}.
\newblock {\em Bioinformatics}, 6(4):309--318, 10 1990.
\newblock \href {https://doi.org/10.1093/bioinformatics/6.4.309}
  {\path{doi:10.1093/bioinformatics/6.4.309}}.

\bibitem{n6}
Kuo-Chung Tai.
\newblock The tree-to-tree correction problem.
\newblock {\em J. ACM}, 26(3):422--433, jul 1979.
\newblock \href {https://doi.org/10.1145/322139.322143}
  {\path{doi:10.1145/322139.322143}}.

\bibitem{apsp}
Ryan Williams.
\newblock Faster all-pairs shortest paths via circuit complexity.
\newblock In {\em Proceedings of the Forty-Sixth Annual ACM Symposium on Theory
  of Computing}, STOC '14, pages 664--673, New York, NY, USA, 2014. Association
  for Computing Machinery.
\newblock \href {https://doi.org/10.1145/2591796.2591811}
  {\path{doi:10.1145/2591796.2591811}}.

\bibitem{n4}
Kaizhong Zhang and Dennis~E. Shasha.
\newblock Simple fast algorithms for the editing distance between trees and
  related problems.
\newblock {\em {SIAM} J. Comput.}, 18(6):1245--1262, 1989.
\newblock \href {https://doi.org/10.1137/0218082} {\path{doi:10.1137/0218082}}.

\end{thebibliography}
